\documentclass[letterpaper, 10 pt, conference]{ieeeconf}  

\IEEEoverridecommandlockouts                              

\usepackage{amsmath,amssymb,amsfonts}

\usepackage{amsthm}
\usepackage{mathtools}
\usepackage{algorithm,algorithmic}
\usepackage{booktabs}
\usepackage{cite}
\usepackage{hyperref}
\usepackage{xcolor}
\usepackage{graphicx}
\usepackage{tikz}
\usetikzlibrary{arrows.meta,positioning,shapes.geometric,fit,backgrounds}
\hypersetup{colorlinks=true,linkcolor=blue,citecolor=blue,urlcolor=blue}

\usepackage{setspace}
\renewcommand{\baselinestretch}{0.99}

\newtheorem{theorem}{Theorem}

\newtheorem{proposition}[theorem]{Proposition}
\newtheorem{corollary}[theorem]{Corollary}
\newtheorem{assumption}{Assumption}
\newtheorem{remark}{Remark}

\newtheorem{problem}{Problem}

\newcommand{\E}{\mathbb{E}}
\newcommand{\R}{\mathbb{R}}
\newcommand{\Spp}{\mathbb{S}_{++}^n}
\newcommand{\Sp}{\mathbb{S}_{+}^n}
\newcommand{\tr}{\mathrm{tr}}
\newcommand{\range}{\mathrm{range}}
\newcommand{\rank}{\mathrm{rank}}
\newcommand{\Normal}{\mathcal{N}}
\newcommand{\Fty}{\mathcal{F}_t^y}
\newcommand{\Var}{\mathrm{Var}}

\title{Path Integral Control in Gaussian Belief Space \\
for Partially Observed Systems}

\author{Goutam Das\textsuperscript{1}, Takashi Tanaka\textsuperscript{2} 
\thanks{This work is supported by DARPA COMPASS program grant HR0011-25-3-0210 and AFOSR DSCT program grant FA9550-25-1-0347.}
\thanks{
All authors are associated with the Networked Control Systems Lab at Purdue University. Emails: \textsuperscript{1} das347@purdue.edu, and \textsuperscript{2} tanaka16@purdue.edu.
}}

\begin{document}

\maketitle

\begin{abstract}
This paper extends path integral control (PIC) to partially observed systems by formulating the problem in Gaussian belief space. PIC uses the matching condition, which requires the diffusion and control channels to be proportional, to linearize the Hamilton--Jacobi--Bellman equation through the Cole--Hopf transform. We show that the belief diffusion satisfies this condition only when the observation function is affine. Thus, the standard Cole--Hopf linearization does not apply to non-affine observation models. Restricting the problem to Gaussian beliefs provides a finite-dimensional approximation with deterministic covariance evolution and reduces the problem to stochastic control of the belief mean. We derive necessary and sufficient conditions for matching in this reduced space, obtain an exact Cole--Hopf linearization and a Feynman--Kac representation, and develop the MPPI-Belief algorithm. Numerical experiments on a navigation task with state-dependent observation noise demonstrate the effectiveness of MPPI-Belief compared with certainty-equivalent and particle-filter-based methods.
\end{abstract}


\section{Introduction}
\label{sec:introduction}

\noindent
Stochastic optimal control under partial observability is central to robotics, autonomous navigation, and aerospace systems~\cite{oguri2024chance,dor2024astroslam,lauri2022partially}. When the state cannot be measured directly, the controller must plan over a \emph{belief}, which is a probability distribution over possible states. This leads to partially observable Markov decision processes (POMDPs) that are generally intractable in continuous state, action, and observation spaces~\cite{platt2010belief,sunberg2018online}.

Path integral control (PIC) provides a powerful framework for \emph{fully observed} stochastic control~\cite{kappen2005path,todorov2006linearly}. Its \emph{matching condition} requires the noise covariance to be proportional to the control authority. When this condition holds, the nonlinear HJB equation can be linearized through the Cole--Hopf transform. The resulting equation admits a Feynman--Kac representation and can be solved using Monte Carlo sampling. This principle underlies model predictive path integral control (MPPI)~\cite{williams2017information} and extensions to state-dependent diffusion~\cite{theodorou2010generalized}, constrained problems~\cite{gandhi2021robust}, and risk-sensitive formulations~\cite{yin2023risk}. These methods assume full state observation.

Under partial observability, the relevant state is the belief rather than the physical state. This raises the question of whether the matching condition can hold in belief space. Existing approaches to partially observed MPPI focus mainly on combining filtering with sampling-based control. Mohamed et~al.~\cite{mohamed2022autonomous} use local costmaps for navigation in unknown environments. Abraham et~al.~\cite{abraham2020model} consider learned uncertain dynamics, while Hoshino et~al.~\cite{hoshino2025path} combine particle filtering with path integral control. Gaussian belief-space methods such as LQG-MP~\cite{van2012motion} and maximum-likelihood belief planning~\cite{platt2010belief} do not connect their Gaussian structure to the PIC matching condition. Thus, whether this condition can hold in belief space remains unresolved. The main distinction of this work is structural. Rather than combining a filter with a path integral rollout controller, we characterize when the observation-driven belief diffusion is compatible with the PIC control channel and when the Cole--Hopf linearization is valid in the reduced belief space.

This paper extends PIC to Gaussian belief space and makes the following contributions.

(i) We show that the observation-driven belief diffusion is tangent to the Gaussian mean manifold if and only if the observation function is affine. Thus, non-affine observation models do not support the matching structure required for the standard Cole--Hopf linearization.

(ii) Under a Gaussian approximation, the belief covariance evolves deterministically, reducing the problem to stochastic control of the belief mean. We derive necessary and sufficient range-space conditions for matching in this reduced space and obtain an exact Cole--Hopf linearization.

(iii) We derive a Feynman--Kac representation, develop the MPPI-Belief controller in Algorithm~\ref{alg:mppi}, and introduce a risk-sensitive cost extension.

The paper is organized as follows. Section~\ref{sec:problem} formulates the control problem. Section~\ref{sec:belief_dynamics} derives the belief dynamics. Section~\ref{sec:gaussian} develops the Gaussian reduction and distinguishes the exact and approximate regimes. Section~\ref{sec:pic_belief} presents the main results. Sections~\ref{sec:algorithm} and~\ref{sec:experiments} present the algorithm and numerical results.

\subsection{Notation}
\noindent
We write $\range(M)$ and $\rank(M)$ for the range and rank of $M$, respectively. The sets $\Sp$ and $\Spp$ denote the cones of $n\times n$ positive semidefinite and positive definite matrices. We use $M^\dagger$ for the Moore--Penrose pseudoinverse and $\Normal(\mu,\Sigma)$ for the Gaussian distribution with mean $\mu$ and covariance $\Sigma$.
\section{Problem Formulation}
\label{sec:problem}

\noindent
We consider a continuous-time control-affine stochastic system whose state is observed through a noisy output. The state and observation equations are
\begin{subequations}\label{eq:system}
\begin{align}
    dx_t &= f(x_t)\,dt + G(x_t)\,u_t\,dt + H\,dw_t,
        \label{eq:dynamics} \\
    dy_t &= c(x_t)\,dt + \sigma_o(x_t)\,d\nu_t,
        \label{eq:observation}
\end{align}
\end{subequations}
where $x_t \in \R^n$ is the state, $u_t \in \R^\ell$ is the control input, and $y_t \in \R^p$ is the observation. The function $f : \R^n \to \R^n$ is the uncontrolled drift, $G : \R^n \to \R^{n \times \ell}$ is the control input matrix, and $c : \R^n \to \R^p$ is the observation function. The matrix $H \in \R^{n \times m}$ determines the process noise channel. The function $\sigma_o : \R^n \to \R^{p \times p}$ is the observation noise matrix and is assumed to be invertible for every $x$. The Brownian motions $w_t \in \R^m$ and $\nu_t \in \R^p$ are mutually independent. We define the process noise covariance and observation noise covariance as
\[
Q := HH^\top \in \Sp,
\qquad
R_o(x) := \sigma_o(x)\sigma_o(x)^\top \in \mathbb{S}_{++}^{p}.
\]
A constant $\sigma_o$ gives the standard observation model, while a state-dependent $\sigma_o$ represents sensing quality that varies with the state. We write $R_o$ for $R_o(x)$ when the dependence is clear.

Let $\Fty := \sigma\{y_s : 0 \leq s \leq t\}$ denote the $\sigma$-algebra generated by the observation history up to time~$t$. An \emph{admissible control} is an $\Fty$-adapted, square-integrable process $\{u_t\}_{t \in [0,T]}$. Thus, $u_t$ may depend on the observation history $\{y_s\}_{0 \leq s \leq t}$ but not directly on the unobserved state $x_t$. The cost functional is
\begin{equation}
    J = \E\left[ \bar{\phi}(x_T) + \int_0^T \left( \bar{q}(x_s)
        + \frac{1}{2} u_s^\top R\, u_s \right) ds \right],
    \label{eq:true_state_objective}
\end{equation}
where $\bar{\phi} : \R^n \to \R$ and $\bar{q} : \R^n \to \R$ are the nonnegative terminal and running state costs, and $R \in \mathbb{S}_{++}^{\ell}$ is the control cost weight.

The partially observed control problem is stated below.

\begin{problem}
\label{prob:main}
Consider the state and observation equations~\eqref{eq:dynamics}--\eqref{eq:observation} and the cost functional~\eqref{eq:true_state_objective}. Find an admissible control $\{u_t\}_{t \in [0,T]}$ that minimizes~\eqref{eq:true_state_objective} using only the observation history $\{y_s\}_{0 \leq s \leq t}$.
\end{problem}

Because the controller cannot observe $x_t$ directly, it must plan using the conditional distribution of the state given the observation history. This distribution is called the \emph{belief}~\cite{krishnamurthy2016partially}. The belief-space formulation applies to general partially observed systems and does not require a Gaussian assumption. In Section~\ref{sec:gaussian}, we introduce a Gaussian approximation to obtain the finite-dimensional model used for the path integral formulation.

Under a Gaussian belief $\pi_t = \Normal(\mu_t, \Sigma_t)$, define the expected running cost as
\begin{equation}
    q(\mu, \Sigma) :=
      \E_{x \sim \Normal(\mu, \Sigma)}[\bar{q}(x)],
\end{equation}
and define the terminal cost $\phi(\mu,\Sigma)$ in the same way. For the quadratic state cost $\bar{q}(x) = \tfrac{1}{2}\|x - x^{\mathrm{ref}}\|_{Q^x}^2$, we obtain
\begin{equation}
    q(\mu, \Sigma)
        = \tfrac{1}{2}\|\mu - x^{\mathrm{ref}}\|_{Q^x}^2
        + \tfrac{1}{2}\tr(Q^x\Sigma).
    \label{eq:gauss_cost}
\end{equation}
The first term measures the cost at the belief mean, while the second penalizes state uncertainty through the covariance. Dropping the trace term gives the certainty-equivalent approach, which plans using the mean alone and omits the cost due to uncertainty.

\begin{remark}[Classical PIC Matching]
\label{rem:pic_full}
For the fully observed system~\eqref{eq:dynamics} with cost~\eqref{eq:true_state_objective}, the \emph{matching condition} is
\[
Q = \lambda\,G(x)\,R^{-1}G(x)^\top
\]
for some $\lambda > 0$. Thus, the process noise covariance must be proportional to the control authority, and both must act along the same directions in the state space. When matching holds, the associated HJB equation admits a Cole--Hopf linearization. This produces a linear PDE with a Feynman--Kac representation and enables a Monte Carlo solution. Without matching, this linearization is unavailable. This structure underlies MPPI~\cite{williams2017information}. When matching holds only approximately, importance sampling corrections~\cite{williams2018robust} and modified cost formulations~\cite{mohamed2022autonomous} can be used.
\end{remark}
\section{Belief Dynamics and Feasibility of Matching}
\label{sec:belief_dynamics}

\noindent
This section derives the belief dynamics for the partially observed system~\eqref{eq:dynamics} and~\eqref{eq:observation} and examines the resulting control and diffusion channels in belief space.

Because the state $x_t$ is not directly observed, the controller must use the conditional distribution
$\pi_t(\cdot) := \mathbb{P}(x_t \in \cdot \mid \Fty)$,
which is a measure-valued stochastic process. Its evolution is governed by the Kushner--Stratonovich equation (KSE)~\cite{kushner1967dynamical}
\begin{multline}
    d\pi_t(x)
        = (\mathcal{L}_t^{u*}\pi_t)(x)\,dt
        + \pi_t(x)\,
            \bigl(c(x)-\pi_t(c)\bigr)^\top
            R_o^{-1}\,d\tilde{y}_t,
    \label{eq:kse}
\end{multline}
where $\mathcal{L}_t^{u*}$ is the Fokker--Planck operator associated with the controlled dynamics~\eqref{eq:dynamics}. The process
$d\tilde{y}_t := dy_t-\pi_t(c)\,dt$
is the \emph{innovation process}, and its whitened form
$R_o^{-1/2}d\tilde{y}_t$ is a standard $\Fty$-Brownian motion~\cite{liptser1977statistics}.

To examine whether the PIC structure in Remark~\ref{rem:pic_full} extends to belief space, we write~\eqref{eq:kse} as
$d\pi=a^u(\pi)\,dt+b(\pi)\,d\tilde{y}$. The corresponding channels are
\begin{subequations}\label{eq:belief_channels}
\begin{align}
    a^u(\pi)(x)
        &= a^0(\pi)(x)
        - \nabla \cdot \bigl(G(x)\,u\,\pi(x)\bigr),
        \label{eq:belief_drift} \\
    b(\pi)(x)
        &= \pi(x)\,
            \bigl(c(x)-\pi(c)\bigr)^\top R_o^{-1}.
        \label{eq:belief_diffusion}
\end{align}
\end{subequations}
The control enters the drift through
$-\nabla\cdot(G(x)\,u\,\pi(x))$, while the diffusion is independent of $u$. 

At a heuristic level, belief-space matching would require the observation-driven diffusion~\eqref{eq:belief_diffusion} to be proportional to the control-induced perturbation~\eqref{eq:belief_drift}. If this held, the belief-space HJB equation would admit a Cole--Hopf linearization and a Feynman--Kac representation. We do not claim a general impossibility result on the full measure-valued state. The following theorem instead addresses this question on the Gaussian mean manifold used in the remainder of the paper.

\begin{theorem}[Affine Observation Characterization]
\label{thm:matching_feasibility}
Consider the system~\eqref{eq:dynamics} and~\eqref{eq:observation} with constant $G \in \R^{n \times \ell}$ and observation function $c : \R^n \to \R^p$. Restrict the belief to the Gaussian family
$\{\pi_\mu=\Normal(\mu,\Sigma) : \mu\in\R^n\}$
with fixed $\Sigma\succ0$. Then the observation-driven diffusion lies in the tangent space of this Gaussian mean manifold if and only if
\begin{equation}
    c(x)-\E_{\pi_\mu}[c(x)]
        = A(\mu)(x-\mu)
    \quad \text{for all } \mu\in\R^n,
    \label{eq:matching_condition}
\end{equation}
for some matrix-valued function
$A(\mu)\in\R^{p\times n}$. Moreover, this condition holds if and only if $c$ is affine, so that
$c(x)=Ax+b$
for some constant matrix $A\in\R^{p\times n}$ and vector $b\in\R^p$.
\end{theorem}

\begin{proof}
For the Gaussian family
$\pi_\mu=\Normal(\mu,\Sigma)$
with fixed $\Sigma\succ0$, the mean-tangent direction is proportional to
$\Sigma^{-1}(x-\mu)\,\pi_\mu(x)$.
The observation-driven diffusion~\eqref{eq:belief_diffusion} is proportional to
$\pi_\mu(x)\bigl(c(x)-\bar{c}(\mu)\bigr)$,
where
$\bar{c}(\mu):=\E_{\pi_\mu}[c(x)]$.
Therefore, the diffusion lies in the Gaussian mean tangent space if and only if
\[
c(x)-\bar{c}(\mu)=A(\mu)(x-\mu)
\]
for some $A(\mu)\in\R^{p\times n}$.

Evaluating this identity at $\mu_1\neq\mu_2$ and subtracting gives
\[
\bar{c}(\mu_2)-\bar{c}(\mu_1)
    = \bigl[A(\mu_1)-A(\mu_2)\bigr]x
    - A(\mu_1)\mu_1+A(\mu_2)\mu_2.
\]
The left-hand side is independent of $x$, so
$A(\mu_1)=A(\mu_2)$. Because $\mu_1$ and $\mu_2$ are arbitrary,
$A(\mu)\equiv A$ is constant. Substitution then gives
\[
c(x)=Ax+\bigl(\bar{c}(\mu)-A\mu\bigr).
\]
Because the left-hand side is independent of $\mu$,
$\bar{c}(\mu)-A\mu$ must be constant. Hence,
$c(x)=Ax+b$. The converse follows directly.
\end{proof}

The fixed-covariance restriction in
Theorem~\ref{thm:matching_feasibility}
isolates the role of the observation function $c$.
Section~\ref{sec:gaussian} considers time-varying covariance, and Proposition~\ref{prop:matching_feasibility} analyzes matching for the resulting time-varying diffusion $D_t$.
For state-dependent $G(x)$, the same difficulty remains because the observation and control channels are not aligned in general.

\begin{remark}
Every scalar linear-Gaussian system with $c(x)=ax+b$ and $G\neq0$ satisfies Theorem~\ref{thm:matching_feasibility}. If $D_t\neq0$, it also satisfies the range condition in Proposition~\ref{prop:matching_feasibility}. Therefore, belief-space PIC applies exactly in this case.
\end{remark}

This result motivates restricting the belief to a finite-dimensional family of distributions.

\section{Gaussian Belief-Space Approximation}
\label{sec:gaussian}

\noindent
We restrict the belief to the Gaussian family
$\pi_t=\Normal(\mu_t,\Sigma_t)$
and derive the mean and covariance dynamics. The deterministic evolution of $\Sigma_t$ enables the finite-dimensional belief-space PIC reduction.

\begin{proposition}[Deterministic Covariance in the Linear Case]
\label{thm:linear_cov}
For the linear system $f(x)=Ax$, $G(x)=B$, and $c(x)=Cx$ in~\eqref{eq:dynamics}--\eqref{eq:observation}, with constant $R_o$, the conditional distribution $\pi_t=\Normal(\mu_t,\Sigma_t)$ is exactly Gaussian, with
\begin{subequations}\label{eq:linear_belief}
\begin{align}
    d\mu_t
        &= (A\mu_t+Bu_t)\,dt+K_t\,d\tilde{y}_t,
        \label{eq:mean_linear_kb} \\
    \dot{\Sigma}_t
        &= A\Sigma_t+\Sigma_tA^\top+Q
        -\Sigma_tC^\top R_o^{-1}C\Sigma_t,
        \label{eq:riccati}
\end{align}
\end{subequations}
where $K_t:=\Sigma_tC^\top R_o^{-1}$ is the Kalman gain and
$d\tilde{y}_t=dy_t-C\mu_t\,dt$
is the innovation process. The Riccati equation~\eqref{eq:riccati} is deterministic and independent of the observations and control input.
\end{proposition}

\begin{proof}
The result follows from Kalman--Bucy filtering theory~\cite{liptser1977statistics}.
\end{proof}

The covariance of the innovation term is
\begin{equation}
    D_t:=K_tR_oK_t^\top
    =\Sigma_tC^\top R_o^{-1}C\Sigma_t.
    \label{eq:Dt_def}
\end{equation}
Let $r_t:=\rank(D_t)$ and choose
$L_t\in\R^{n\times r_t}$
such that $D_t=L_tL_t^\top$. The case $D_t=0$ means that the observations do not update the belief mean.

Since the quadratic variation of $K_t\,d\tilde{y}_t$ is $D_t\,dt$, the mean dynamics can be written as
\begin{equation}
    d\mu_t=(A\mu_t+Bu_t)\,dt+L_t\,d\beta_t,
    \label{eq:mean_linear}
\end{equation}
where $\beta_t\in\R^{r_t}$ is a standard $\Fty$-Brownian motion. Since $\Sigma_t$ is deterministic, $D_t$ and $L_t$ are deterministic functions of time.

\begin{remark}[EKF Approximation for Nonlinear Systems]
\label{rem:nonlinear_cov}
For nonlinear systems~\eqref{eq:dynamics}--\eqref{eq:observation}, linearizing $f$ and $c$ about a nominal trajectory $\bar{\mu}_t$, with
$A_t:=\partial f/\partial x|_{\bar{\mu}_t}$
and
$C_t:=\partial c/\partial x|_{\bar{\mu}_t}$,
gives
\begin{equation}
    \dot{\Sigma}_t
        =A_t\Sigma_t+\Sigma_tA_t^\top+Q
        -\Sigma_tC_t^\top R_{o,t}^{-1}C_t\Sigma_t,
    \label{eq:ekf_riccati}
\end{equation}
where $R_{o,t}=R_o$ for constant observation noise and
$R_{o,t}:=R_o(\bar{\mu}_t)$
for state-dependent observation noise. Once $\bar{\mu}_t$ is specified, this equation is deterministic. The residual from linearizing $f$ and $c$ is second order in
$\|\mu_t-\bar{\mu}_t\|$. With state-dependent observation noise, freezing $R_o$ along $\bar{\mu}_t$ is a separate local approximation that introduces a first-order error. These approximations are accurate when the belief remains near the nominal trajectory. A rigorous bound is beyond the scope of this paper.
\end{remark}

\begin{assumption}[Deterministic Gain Schedule]
\label{assum:gain_schedule}
Over $[0,T]$, $K_t$, $\Sigma_t$, and
$D_t:=K_tR_{o,t}K_t^\top$
are predetermined deterministic functions of time obtained from~\eqref{eq:riccati} or~\eqref{eq:ekf_riccati}. For nonlinear systems, we set
$\bar{G}_t:=G(\bar{\mu}_t)$. For linear systems,
$\bar{G}_t=B$ and $R_{o,t}=R_o$. With state-dependent observation noise,
$R_{o,t}:=R_o(\bar{\mu}_t)$
is not recomputed at sampled states. This freezing is distinct from the EKF linearization and keeps $\Sigma_t$, $K_t$, and $D_t$ deterministic.
\end{assumption}

Under Assumption~\ref{assum:gain_schedule}, the belief-mean dynamics are approximated by
\begin{equation}
    d\mu_t=\bigl(f(\mu_t)+\bar{G}_t\,u_t\bigr)\,dt
        +L_t\,d\beta_t,
    \label{eq:mean_nonlinear}
\end{equation}
where $D_t=L_tL_t^\top$ is deterministic. Since $\Sigma_t$ is known, its contribution to the cost is deterministic. We define
$q_t(\mu):=q(\mu,\Sigma(t))$
and
$\phi_T(\mu):=\phi(\mu,\Sigma(T))$.

Although the KSE allows general beliefs, the finite-dimensional HJB and PIC formulations require a Gaussian belief with deterministic covariance. They do not directly apply to particle, mixture, or other non-Gaussian beliefs.

\section{Path Integral Control in Gaussian Belief Space}
\label{sec:pic_belief}

\noindent
This section derives the matching condition, Cole--Hopf linearization, and Feynman--Kac representation for the Gaussian belief mean. Unlike fully observed PIC, the diffusion $D_t$ is generated by the observation channel rather than the process noise $Q$.

\subsection{HJB Equation for the Belief Mean}
\label{sec:hjb}

\noindent
Under Assumption~\ref{assum:gain_schedule}, the deterministic covariance reduces the belief-space problem to finite-dimensional stochastic control of the mean. Let $V(t,\mu)$ denote the optimal cost-to-go. It satisfies
\begin{multline}
    -\partial_t V
        = \min_u \biggl\{
            q_t(\mu)
            + \tfrac{1}{2}u^\top R\,u \\
        + (\nabla_\mu V)^\top
            \bigl(f(\mu)+\bar{G}_t\,u\bigr)
        + \tfrac{1}{2}\tr(D_t\,\nabla_{\mu\mu}^2 V)
        \biggr\}
    \label{eq:hjb}
\end{multline}
with $V(T,\mu)=\phi_T(\mu)$. First-order optimality gives
$u^*=-R^{-1}\bar{G}_t^\top\nabla_\mu V$.
Substitution into~\eqref{eq:hjb} gives
\begin{multline}
    -\partial_t V
        =q_t(\mu)
        +(\nabla_\mu V)^\top f(\mu)
        +\tfrac{1}{2}\tr(D_t\,\nabla_{\mu\mu}^2 V) \\
        -\tfrac{1}{2}(\nabla_\mu V)^\top
            \bar{G}_tR^{-1}\bar{G}_t^\top\nabla_\mu V.
    \label{eq:hjb_opt}
\end{multline}
The quadratic gradient term makes~\eqref{eq:hjb_opt} nonlinear. The matching condition allows this term to be removed by the Cole--Hopf transform.

\subsection{Matching Condition and Feasibility}
\label{sec:matching_gaussian}

\noindent
The Cole--Hopf transform requires the belief-mean diffusion to be proportional to the scheduled control authority.

\begin{assumption}[Scheduled Matching]
\label{assum:matching}
For each $t\in[0,T]$, the belief-mean diffusion matrix $D_t$ and scheduled control matrix $\bar{G}_t$ satisfy
\begin{equation}
    D_t=\lambda\,\bar{G}_tR^{-1}\bar{G}_t^\top,
    \label{eq:matching}
\end{equation}
where $\lambda>0$ is the path integral temperature.
\end{assumption}

\begin{proposition}[Matching Feasibility]
\label{prop:matching_feasibility}
For each $t\in[0,T]$, the following statements hold.

\emph{(i)} There exists a symmetric matrix $W\succeq0$ such that
$D_t=\lambda\,\bar{G}_tW\bar{G}_t^\top$
if and only if
$\range(D_t)\subseteq\range(\bar{G}_t)$.

\emph{(ii)} There exists $R\succ0$ satisfying~\eqref{eq:matching} if and only if
$\range(D_t)=\range(\bar{G}_t)$.
\end{proposition}

\begin{proof}
For part~(i),
$\range(\bar{G}_tW\bar{G}_t^\top)\subseteq\range(\bar{G}_t)$
for any $W$, so the range inclusion is necessary. Conversely, if
$\range(D_t)\subseteq\range(\bar{G}_t)$, then
$D_t^{1/2}=\bar{G}_tP$
for some matrix $P$. Setting
$W=PP^\top/\lambda$
gives the result.

For part~(ii), if $R\succ0$, then $R^{-1}\succ0$, so
\[
\range(\bar{G}_tR^{-1}\bar{G}_t^\top)
    =\range(\bar{G}_t),
\]
which forces range equality. Conversely, if
$\range(D_t)=\range(\bar{G}_t)$,
an SPD matrix $R^{-1}$ can be constructed on
$\range(\bar{G}_t)$
using the SVD and extended to the full input space by adding a positive multiple of the projector onto
$\ker(\bar{G}_t)$.
\end{proof}

Here $R\succ0$ is the fixed control cost weight from~\eqref{eq:true_state_objective}. Proposition~\ref{prop:matching_feasibility} gives instantaneous feasibility at each $t$. Assumption~\ref{assum:matching} is stronger because one $R$ and one $\lambda$ must satisfy~\eqref{eq:matching} for every $t\in[0,T]$. The matrix $W$ in part~(i) is an auxiliary existence variable rather than a design parameter.

\subsection{Cole--Hopf Linearization}
\label{sec:cole_hopf}

\noindent
Define the \emph{desirability function}
$\Psi(t,\mu):=\exp(-V(t,\mu)/\lambda)$,
so that
$V(t,\mu)=-\lambda\log\Psi(t,\mu)$.

\begin{theorem}[Exact Linearization via Cole--Hopf]
\label{thm:linearization}
Under Assumptions~\ref{assum:gain_schedule} and~\ref{assum:matching}, the substitution
$V(t,\mu)=-\lambda\log\Psi(t,\mu)$
converts~\eqref{eq:hjb_opt} into the linear PDE
\begin{equation}
    -\partial_t\Psi
        =-\frac{q_t(\mu)}{\lambda}\Psi
        +(\nabla_\mu\Psi)^\top f(\mu)
        +\tfrac{1}{2}\tr(D_t\,\nabla_{\mu\mu}^2\Psi),
    \label{eq:linear_pde}
\end{equation}
with terminal condition
$\Psi(T,\mu)=\exp(-\phi_T(\mu)/\lambda)$.
\end{theorem}

\begin{proof}
From $V=-\lambda\log\Psi$ we have $\nabla_\mu V=-(\lambda/\Psi)\nabla_\mu\Psi$ and $\nabla_{\mu\mu}^2V=-(\lambda/\Psi)\nabla_{\mu\mu}^2\Psi+(\lambda/\Psi^2)(\nabla_\mu\Psi)(\nabla_\mu\Psi)^\top$. Substituting into~\eqref{eq:hjb_opt}, the Hessian trace contributes $(\lambda/2\Psi^2)(\nabla_\mu\Psi)^\top D_t\nabla_\mu\Psi$ and the quadratic control term contributes $-(\lambda^2/2\Psi^2)(\nabla_\mu\Psi)^\top\bar{G}_tR^{-1}\bar{G}_t^\top\nabla_\mu\Psi$. These cancel under~\eqref{eq:matching}, giving~\eqref{eq:linear_pde}.
\end{proof}

Since~\eqref{eq:linear_pde} is linear, the desirability admits a stochastic representation.

\begin{corollary}[Belief-Space Feynman--Kac Representation]
\label{cor:feynman_kac}
Under Assumptions~\ref{assum:gain_schedule} and~\ref{assum:matching}, let $\mathbb{Q}$ denote the probability measure associated with the uncontrolled belief-mean dynamics. Then
\begin{equation}
    \Psi(t,\mu)
        =\E^{\mathbb{Q}}\biggl[
            \exp\!\biggl(
                -\frac{1}{\lambda}\int_t^Tq_s(\mu_s)\,ds
                -\frac{\phi_T(\mu_T)}{\lambda}
            \biggr)
            \,\bigg|\,\mu_t=\mu
        \biggr].
    \label{eq:feynman_kac}
\end{equation}
\end{corollary}

This representation expresses the desirability as an expectation over uncontrolled belief-mean trajectories and enables Monte Carlo approximation.

\begin{remark}[Risk-Sensitive Extension]
\label{rem:risk_sensitive}
The expected cost
$q(\mu,\Sigma)=\E[\bar{q}(x)]$
can be replaced by
$q_{\mathrm{RS}}(\mu,\Sigma)
=\E[\bar{q}(x)]+\frac{\theta}{2}\Var[\bar{q}(x)]$,
which is the second-order cumulant approximation of the entropic risk measure~\cite{jacobson1973optimal,whittle1990risk}. The variance term penalizes beliefs that assign significant probability to high-cost states even when the mean trajectory appears safe. The parameter $\theta>0$ controls this penalty, and $\theta=0$ recovers the risk-neutral cost. The PIC structure is unchanged because only the running cost is modified.
\end{remark}

\section{MPPI-Belief Algorithm}
\label{sec:algorithm}

\noindent
This section develops a discrete-time Monte Carlo approximation of the belief-space controller and presents the MPPI-Belief algorithm.

Under an Euler--Maruyama discretization of the uncontrolled belief-mean dynamics~\eqref{eq:mean_nonlinear} with $u_t\equiv0$ at times $t_k=k\Delta t$, the transitions are
\begin{equation}
    \mu_{k+1}
        =\mu_k+f(\mu_k)\,\Delta t+L_k\,\epsilon_k,
    \qquad
    \epsilon_k\sim\Normal(0,I_{r_k}\,\Delta t),
    \label{eq:discrete_dynamics}
\end{equation}
where $L_kL_k^\top=D_{t_k}$ and
$r_k:=\rank(D_{t_k})$.

\begin{theorem}[Discrete-Time Controller]
\label{thm:discrete_control}
Under Assumptions~\ref{assum:gain_schedule}--\ref{assum:matching}, suppose the coefficients are frozen over the first time step, as is standard in MPPI~\cite{williams2017information}. The controller at step $k=0$ is
\begin{equation}
    u_0^*
        =\sum_{i=1}^N\tilde{w}^{(i)}
        \frac{
            R^{-1}\bar{G}_0^\top
            (\bar{G}_0R^{-1}\bar{G}_0^\top)^\dagger
            L_0\epsilon_0^{(i)}
        }{\Delta t},
    \label{eq:discrete_control}
\end{equation}
where $\bar{G}_0:=G(\bar{\mu}_0)$,
$\tilde{w}^{(i)}:=w^{(i)}/\sum_jw^{(j)}$,
$w^{(i)}:=\exp(-S^{(i)}/\lambda)$,
and $S^{(i)}$ is the path cost of trajectory $i$. For rank-deficient $D_0$, the pseudoinverse restricts the score to $\range(D_0)$ and ensures that the controller is well defined for underactuated systems.
\end{theorem}

\begin{proof}
From the first-order optimality condition
$u^*=-R^{-1}\bar{G}_t^\top\nabla_\mu V$
and
$\nabla_\mu V=-(\lambda/\Psi)\nabla_\mu\Psi$,
\[
u^*
    =\lambda R^{-1}\bar{G}_t^\top
      \frac{\nabla_\mu\Psi}{\Psi}.
\]
By the Feynman--Kac representation in Corollary~\ref{cor:feynman_kac},
$\nabla_{\mu_0}\log\Psi$
can be written as a cost-weighted expectation of the score function of the uncontrolled transition density. For the discrete transition~\eqref{eq:discrete_dynamics} with covariance $D_0\Delta t$, the score restricted to $\range(D_0)$ is
$D_0^\dagger L_0\epsilon_0/\Delta t$.
Using the matching condition
$D_0=\lambda\bar{G}_0R^{-1}\bar{G}_0^\top$
gives
\[
\lambda R^{-1}\bar{G}_0^\top D_0^\dagger
    =
R^{-1}\bar{G}_0^\top
(\bar{G}_0R^{-1}\bar{G}_0^\top)^\dagger.
\]
Approximating the resulting expectation using $N$ Monte Carlo samples yields~\eqref{eq:discrete_control}.
\end{proof}

\begin{algorithm}[t]
\caption{MPPI-Belief}
\label{alg:mppi}
\begin{algorithmic}[1]
\REQUIRE Belief $(\mu_0,\Sigma_0)$, number of samples $N$,
    horizon length $H=T/\Delta t$, step size $\Delta t$,
    temperature $\lambda$
\STATE Precompute $\{\Sigma_k,K_k,D_k,L_k\}_{k=0}^{H}$
    by solving the Riccati equation along the nominal
    trajectory $\bar{\mu}$
\STATE $\bar{G}_0\leftarrow G(\bar{\mu}_0)$
\STATE $M_0\leftarrow R^{-1}\bar{G}_0^\top
    (\bar{G}_0R^{-1}\bar{G}_0^\top)^\dagger L_0$
\FOR{$i=1,\ldots,N$}
    \STATE $\mu_0^{(i)}\leftarrow\mu_0$, \;
        $S^{(i)}\leftarrow0$
    \FOR{$k=0,\ldots,H-1$}
        \STATE Sample
        $\epsilon_k^{(i)}\sim\Normal(0,I_{r_k}\Delta t)$
        \STATE $\mu_{k+1}^{(i)}\leftarrow\mu_k^{(i)}
            +f(\mu_k^{(i)})\Delta t
            +L_k\epsilon_k^{(i)}$
        \STATE $S^{(i)}\leftarrow S^{(i)}
            +q_{t_k}(\mu_k^{(i)})\Delta t$
    \ENDFOR
    \STATE $S^{(i)}\leftarrow S^{(i)}
        +\phi_T(\mu_H^{(i)})$
    \STATE $w^{(i)}\leftarrow\exp(-S^{(i)}/\lambda)$
\ENDFOR
\STATE $\tilde{w}^{(i)}\leftarrow
    w^{(i)}/\sum_jw^{(j)}$ for all $i$
\RETURN $u_0^*=M_0
    \sum_i\tilde{w}^{(i)}\epsilon_0^{(i)}/\Delta t$
\end{algorithmic}
\end{algorithm}

In practice, Algorithm~\ref{alg:mppi} is applied in a receding-horizon fashion. At each time step, the belief $(\mu_0,\Sigma_0)$ is updated using the latest filter estimate, and only $u_0^*$ is applied before replanning. Unlike particle-filter-based methods such as the path integral particle filter (PIPF)~\cite{hoshino2025path}, which combine filtering particles with separate path integral controllers, MPPI-Belief uses a single Gaussian belief representation and samples $N$ belief-mean trajectories.

When Assumption~\ref{assum:matching} does not hold, Algorithm~\ref{alg:mppi} retains the update~\eqref{eq:discrete_control}, which projects the sampled score onto the control channel. The exact score of the sampled transition is $D_0^\dagger L_0\epsilon_0/\Delta t$, and the two coincide under matching.

\section{Numerical Validation}
\label{sec:experiments}

\noindent
We consider a 2D double-integrator with state
$x=[p_x,p_y,v_x,v_y]^\top$, control $u\in\R^2$, and position observations
$C=[I_2\;0_{2\times2}]$. The observation-noise standard deviation is
$\sigma_o(p_x)=|p_x-x_{\mathrm{light}}|/\sqrt{2}+0.1$,
where $x_{\mathrm{light}}=5$\,m defines the most informative region. Because the noise is position dependent, each planning update evaluates $R_o$ along the nominal trajectory and fixes the resulting $\Sigma_k$, $K_k$, and $D_k$ schedules during rollouts, as required by Assumption~\ref{assum:gain_schedule}, rather than the exact linear-Gaussian regime of Proposition~\ref{thm:linear_cov}.

Two radius-$0.65$\,m obstacles at $(3,\pm1)$\,m form a corridor in a high-noise region, and $\sigma_w=0.30$. MPPI-Belief and CE-MPPI use $N=500$ samples. All methods use $H=30$ and $\Delta t=0.1$\,s, while the sampling methods use $\lambda=1$. The $3$\,s horizon covers the corridor, and the $18$\,ms update time is below the $0.1$\,s control period. In general, $H$ should cover the nearest safety-relevant decision, $N$ should respect the control period, and $\lambda$ should prevent weight collapse, monitored through
$\mathrm{ESS}=1/\sum_i(\tilde{w}^{(i)})^2$.

We compare MPPI-Belief with $\theta=1$ against EKF-iLQG~\cite{todorov2005generalized}, CE-MPPI with $\theta=0$, and PIPF~\cite{hoshino2025path} with $K=50$ particles and $L=200$ samples. CE-MPPI uses the risk-neutral expected cost without the variance penalty. All methods share the base state, obstacle, and control costs and the same horizon. The sampling methods also share the same temperature.\footnote{Our PIPF baseline uses a bootstrap particle filter and CE-MPPI importance weights rather than the KL-optimal proposal of~\cite{hoshino2025path}. It is therefore a simplified PIPF-style baseline.}

\begin{table}[b]
\centering
\small
\caption{Performance over 200 Monte Carlo trials.}
\label{tab:exp_baseline}
\begin{tabular}{lccc}
\toprule
Method & Cost & Collision & Time (ms) \\
\midrule
EKF-iLQG    & $46.1\pm6.2$ & 24.5\% & 148 \\
CE-MPPI     & $50.5\pm6.2$ & 23.0\% & 13 \\
PIPF        & $47.6\pm5.5$ & 18.0\% & 296 \\
MPPI-Belief & $56.5\pm5.2$ & \textbf{0.0\%} & 18 \\
\bottomrule
\end{tabular}
\end{table}

Table~\ref{tab:exp_baseline} and Fig.~\ref{fig:exp_trajectories} show zero collisions for MPPI-Belief, compared with 18.0--24.5\% for the baselines. Its risk-sensitive cost favors the informative region before the corridor.

To quantify the violation of~\eqref{eq:matching}, define
$\Gamma_k:=\bar{G}_kR^{-1}\bar{G}_k^\top$
 and
\begin{equation}
    \varepsilon_{\mathrm{match}}
    :=\min_{\lambda>0}
    \frac{\left(\sum_{k=0}^{H-1}
        \|D_k-\lambda\Gamma_k\|_F^2\right)^{1/2}}
         {\left(\sum_{k=0}^{H-1}
        \|D_k\|_F^2\right)^{1/2}}.
    \label{eq:matching_residual}
\end{equation}
The residual is zero exactly when one $\lambda$ satisfies~\eqref{eq:matching} at every prediction step. Here
$\range(D_k)=\range(\Sigma_kC^\top)$
contains position directions, whereas $\range(\bar{G}_k)$ is the velocity subspace. The range condition in Proposition~\ref{prop:matching_feasibility}(i) therefore fails, and exact matching is impossible. We measure $\varepsilon_{\mathrm{match}}=0.99$ for this schedule, confirming that the control and observation channels are nearly orthogonal here.
Thus, the experiment evaluates robustness to matching error rather than the exactly linearizable case.

\begin{figure}[t]
    \centering
    \includegraphics[width=\columnwidth]
        {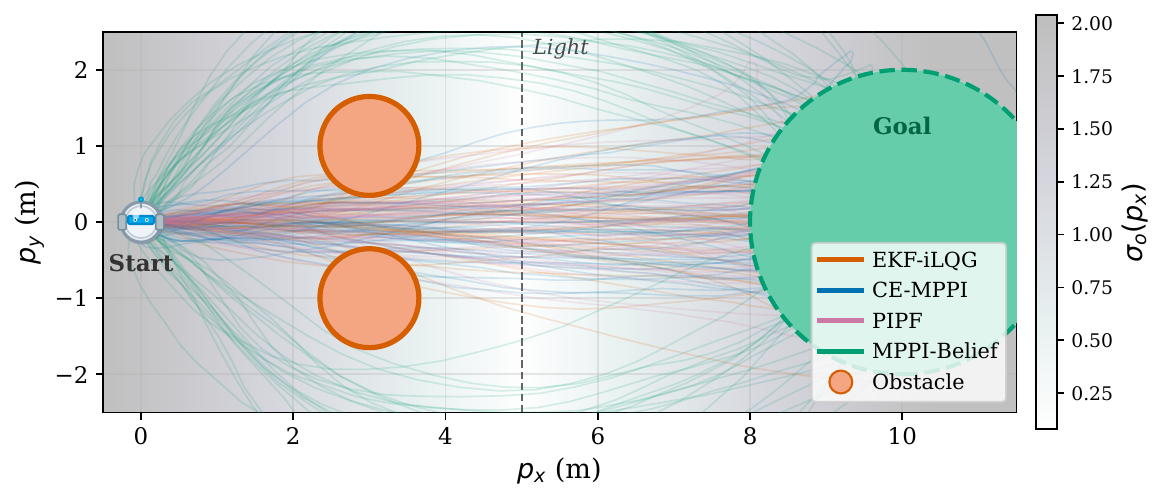}
    \caption{Fifty trajectories per method. Shading denotes $\sigma_o(p_x)$, which is minimized near $p_x=5$\,m.}
    \label{fig:exp_trajectories}
\end{figure}

\begin{figure}[t]
    \centering
    \includegraphics[width=\columnwidth]
        {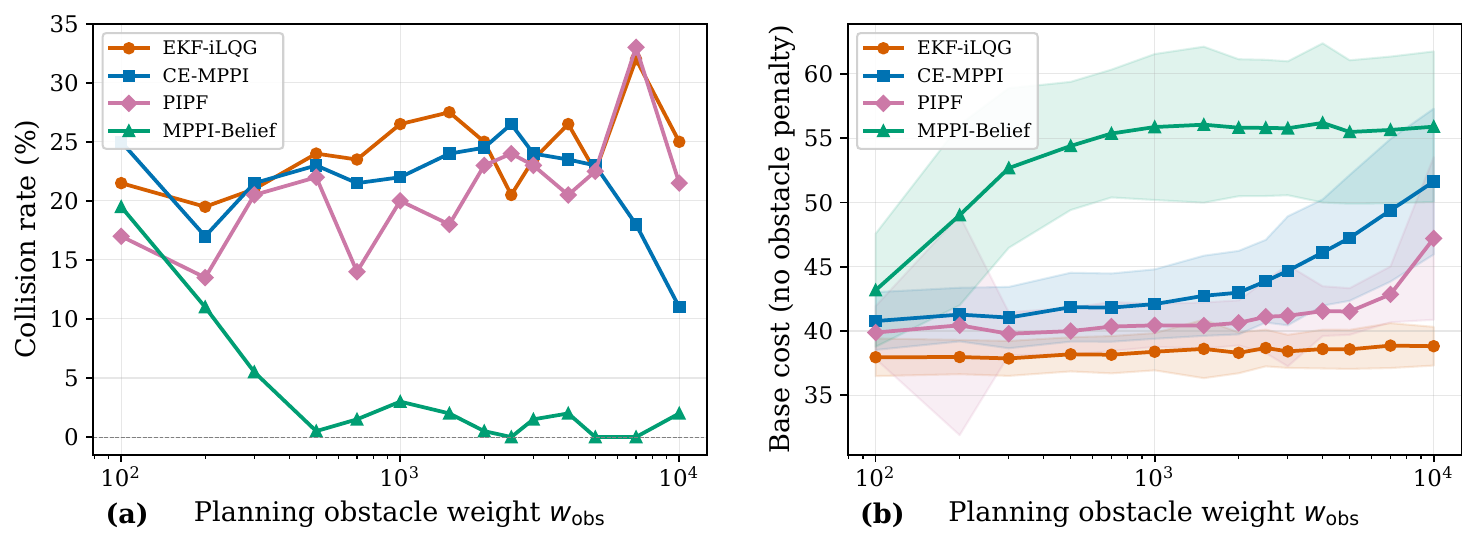}
    \caption{Obstacle-weight sweep over 200 trials. (a) Collision rate. (b) Base cost excluding obstacle penalties.}
    \label{fig:exp_fairness}
\end{figure}

Figure~\ref{fig:exp_fairness} varies
$w_{\mathrm{obs}}\in[10^2,10^4]$.
Baseline collision rates remain 11--33\%, while MPPI-Belief falls below 1\% at $w_{\mathrm{obs}}=500$ and reaches zero at 2500. Thus, obstacle-weight tuning alone does not explain the safety improvement.

Finally, we vary
$\theta\in\{0,0.5,1,2,5\}$
over 200 trials per setting. Increasing $\theta$ from 0 to 0.5 reduces the collision rate from 23.0\% to 2.5\% and increases mean clearance from $0.11$ to $0.80$\,m. At $\theta=1$, MPPI-Belief achieves zero collisions and $0.82$\,m clearance, while the cost increases from $50.5$ to $56.5$. For $\theta\geq2$, collisions increase slightly to 0.5--1.5\%. Therefore, $\theta\in[0.5,1]$ gives the best safety tradeoff in this experiment.

Changing $\theta$ preserves the $O(NH)$ rollout complexity but can concentrate the importance weights. A practical choice is the smallest value that satisfies the safety requirement while maintaining a nondegenerate effective sample size. These finite-horizon results do not guarantee closed-loop stability.

\section{Conclusion}
\label{sec:conclusion}

\noindent
This paper developed path integral control for partially observed systems using a Gaussian belief-space approximation. We showed that compatibility of the observation-driven diffusion with the Gaussian mean manifold requires affine observations and characterized matching in the reduced space. Under scheduled matching, this gives an exact path integral formulation and the MPPI-Belief algorithm. Numerical experiments showed improved safety over certainty-equivalent and particle-filter-based baselines at practical computational cost. The nonlinear formulation relies on an EKF schedule, and the validation is limited to one low-dimensional example without a stability guarantee. Future work will quantify performance under matching error and address non-Gaussian beliefs, closed-loop stability, and hardware validation.


\bibliographystyle{IEEEtran}
\bibliography{references}

\end{document}